\documentclass[11pt]{article}

\usepackage[a4paper,margin=1in]{geometry}
\usepackage[T1]{fontenc}
\usepackage{amsmath,amssymb,amsthm,mathtools}
\usepackage{bm}
\usepackage{booktabs}
\usepackage{array}
\usepackage{xcolor}
\usepackage{pgfplots}
\pgfplotsset{compat=1.18}
\usepackage[numbers,sort&compress]{natbib}
\usepackage[colorlinks=true,linkcolor=blue!55!black,citecolor=blue!55!black,urlcolor=blue!55!black]{hyperref}
\usepackage[capitalise,noabbrev]{cleveref}

\theoremstyle{plain}
\newtheorem{theorem}{Theorem}
\newtheorem{proposition}{Proposition}
\newtheorem{lemma}{Lemma}
\newtheorem{corollary}{Corollary}
\theoremstyle{definition}
\newtheorem{assumption}{Assumption}

\theoremstyle{remark}
\newtheorem{remark}{Remark}

\crefname{assumption}{Assumption}{Assumptions}
\Crefname{assumption}{Assumption}{Assumptions}

\newcommand{\R}{\mathbb{R}}
\newcommand{\E}{\mathbb{E}}
\newcommand{\Prob}{\mathbb{P}}
\newcommand{\F}{\mathcal{F}}
\newcommand{\Adm}{\mathcal{A}}
\newcommand{\Lop}{\mathcal{L}}
\newcommand{\fmax}{f_{\max}}
\newcommand{\dd}{\,\mathrm{d}}

\newcommand{\dlog}{\mathrm{d}\!\log}
\DeclareMathOperator*{\argmax}{arg\,max}

\definecolor{navy}{HTML}{1F3B6F}
\definecolor{teal}{HTML}{02818A}
\definecolor{midblue}{HTML}{3690C0}
\definecolor{lightblue}{HTML}{A6CEE3}
\definecolor{salmon}{HTML}{FB8072}
\definecolor{greyln}{HTML}{BDBDBD}

\def\DATAfeecurve{(14.14,25.2) (20,25.54) (24.49,25.97) (28.28,26.45) (31.62,26.97) (34.64,27.53) (37.42,28.11) (40,28.71) (42.43,29.33) (44.72,29.96) (46.9,30.61) (48.99,31.27) (50.99,31.94) (52.92,32.63) (54.77,33.32) (56.57,34.03) (58.31,34.75) (60,35.48) (61.64,36.22) (63.25,36.98) (64.81,37.75) (66.33,38.53) (67.82,39.33) (69.28,40.14) (70.71,40.98) (72.11,41.83) (73.48,42.7) (74.83,43.59) (76.16,44.51) (77.46,45.45) (78.74,46.42) (80,47.43) (81.24,48.46) (82.46,49.54) (83.67,50.66) (84.85,51.83) (86.02,53.06) (87.18,54.35) (88.32,55.71) (89.44,57.17) (90.55,58.73) (91.65,60.43) (92.74,62.3) (93.81,64.4) (94.87,66.83) (95.92,69.8) (96.95,73.87) (97.98,100) (98.99,100) (100,100) (101,100) (102,100) (103,100) (103.9,100) (104.9,100) (105.8,100) (106.8,100) (107.7,100) (108.6,100) (109.5,100)}
\def\DATAlaffer{25}
\def\DATAobjA{(1,-0.5607) (3.106,-0.1139) (5.213,0.1129) (7.319,0.2663) (9.426,0.3801) (11.53,0.467) (13.64,0.5338) (15.74,0.5847) (17.85,0.6225) (19.96,0.6494) (22.06,0.6673) (24.17,0.6776) (26.28,0.6815) (28.38,0.68) (30.49,0.6741) (32.6,0.6645) (34.7,0.6519) (36.81,0.6369) (38.91,0.6198) (41.02,0.6013) (43.13,0.5816) (45.23,0.561) (47.34,0.5398) (49.45,0.5183) (51.55,0.4967) (53.66,0.4751) (55.77,0.4536) (57.87,0.4324) (59.98,0.4117) (62.09,0.3913) (64.19,0.3715) (66.3,0.3522) (68.4,0.3336) (70.51,0.3156) (72.62,0.2983) (74.72,0.2816) (76.83,0.2656) (78.94,0.2503) (81.04,0.2356) (83.15,0.2216) (85.26,0.2083) (87.36,0.1956) (89.47,0.1836) (91.57,0.1721) (93.68,0.1612) (95.79,0.151) (97.89,0.1412) (100,0.132)}
\def\DATAobjAopt{(26.71,0.6816)}
\def\DATAobjB{(1,-2.065) (3.106,-1.069) (5.213,-0.583) (7.319,-0.2803) (9.426,-0.06985) (11.53,0.08478) (13.64,0.2016) (15.74,0.2909) (17.85,0.3592) (19.96,0.4109) (22.06,0.4492) (24.17,0.4768) (26.28,0.4954) (28.38,0.5067) (30.49,0.5119) (32.6,0.512) (34.7,0.5081) (36.81,0.5007) (38.91,0.4906) (41.02,0.4783) (43.13,0.4643) (45.23,0.4489) (47.34,0.4325) (49.45,0.4154) (51.55,0.3977) (53.66,0.3798) (55.77,0.3618) (57.87,0.3439) (59.98,0.3261) (62.09,0.3085) (64.19,0.2913) (66.3,0.2745) (68.4,0.2582) (70.51,0.2424) (72.62,0.2271) (74.72,0.2123) (76.83,0.1982) (78.94,0.1846) (81.04,0.1716) (83.15,0.1592) (85.26,0.1474) (87.36,0.1361) (89.47,0.1254) (91.57,0.1153) (93.68,0.1057) (95.79,0.09657) (97.89,0.08797) (100,0.07985)}
\def\DATAobjBopt{(31.6,0.5125)}
\def\DATAobjC{(1,-4.537) (3.106,-2.816) (5.213,-1.923) (7.319,-1.366) (9.426,-0.9811) (11.53,-0.7002) (13.64,-0.4878) (15.74,-0.3236) (17.85,-0.1951) (19.96,-0.09387) (22.06,-0.01408) (24.17,0.04859) (26.28,0.0974) (28.38,0.1349) (30.49,0.163) (32.6,0.1835) (34.7,0.1976) (36.81,0.2065) (38.91,0.211) (41.02,0.2119) (43.13,0.2098) (45.23,0.2054) (47.34,0.1991) (49.45,0.1912) (51.55,0.1822) (53.66,0.1722) (55.77,0.1616) (57.87,0.1505) (59.98,0.1391) (62.09,0.1275) (64.19,0.1159) (66.3,0.1044) (68.4,0.09308) (70.51,0.08195) (72.62,0.07109) (74.72,0.06056) (76.83,0.05038) (78.94,0.04058) (81.04,0.03117) (83.15,0.02218) (85.26,0.0136) (87.36,0.005446) (89.47,-0.002291) (91.57,-0.009612) (93.68,-0.01652) (95.79,-0.02303) (97.89,-0.02915) (100,-0.03489)}
\def\DATAobjCopt{(40.56,0.2119)}
\def\DATAobjD{(1,-7.992) (3.106,-5.432) (5.213,-4.01) (7.319,-3.097) (9.426,-2.459) (11.53,-1.989) (13.64,-1.63) (15.74,-1.349) (17.85,-1.125) (19.96,-0.9452) (22.06,-0.7987) (24.17,-0.6789) (26.28,-0.5808) (28.38,-0.5003) (30.49,-0.4342) (32.6,-0.38) (34.7,-0.3359) (36.81,-0.3) (38.91,-0.2711) (41.02,-0.248) (43.13,-0.2298) (45.23,-0.2157) (47.34,-0.205) (49.45,-0.1972) (51.55,-0.1917) (53.66,-0.1882) (55.77,-0.1863) (57.87,-0.1857) (59.98,-0.1861) (62.09,-0.1874) (64.19,-0.1894) (66.3,-0.1918) (68.4,-0.1946) (70.51,-0.1977) (72.62,-0.201) (74.72,-0.2043) (76.83,-0.2076) (78.94,-0.2109) (81.04,-0.2142) (83.15,-0.2173) (85.26,-0.2202) (87.36,-0.223) (89.47,-0.2257) (91.57,-0.2281) (93.68,-0.2303) (95.79,-0.2324) (97.89,-0.2342) (100,-0.2358)}
\def\DATAobjDopt{(57.93,-0.1857)}
\def\DATApathA{(0,1) (0.027778,0.99966) (0.055556,0.99931) (0.083333,0.99896) (0.11111,0.9986) (0.13889,0.99826) (0.16667,0.99791) (0.19444,0.99757) (0.22222,0.99721) (0.25,0.99685) (0.27778,0.99648) (0.30556,0.99611) (0.33333,0.99573) (0.36111,0.99535) (0.38889,0.99498) (0.41667,0.9946) (0.44444,0.99422) (0.47222,0.99385) (0.5,0.99347) (0.52778,0.99308) (0.55556,0.99271) (0.58333,0.99234) (0.61111,0.99197) (0.63889,0.99161) (0.66667,0.99125) (0.69444,0.9909) (0.72222,0.99054) (0.75,0.99016) (0.77778,0.98979) (0.80556,0.9894) (0.83333,0.98902) (0.86111,0.98863) (0.88889,0.98824) (0.91667,0.98786) (0.94444,0.98747) (0.97222,0.9871) (1,0.98673) (1.0278,0.98636) (1.0556,0.98599) (1.0833,0.98562) (1.1111,0.98524) (1.1389,0.98486) (1.1667,0.98448) (1.1944,0.9841) (1.2222,0.98373) (1.25,0.98336) (1.2778,0.98299) (1.3056,0.98262) (1.3333,0.98226) (1.3611,0.98191) (1.3889,0.98156) (1.4167,0.9812) (1.4444,0.98085) (1.4722,0.98051) (1.5,0.98016)}
\def\DATApathB{(0,1) (0.027778,1.0001) (0.055556,1.0002) (0.083333,1.0003) (0.11111,1.0004) (0.13889,1.0005) (0.16667,1.0006) (0.19444,1.0007) (0.22222,1.0008) (0.25,1.0008) (0.27778,1.0009) (0.30556,1.001) (0.33333,1.0011) (0.36111,1.0012) (0.38889,1.0013) (0.41667,1.0013) (0.44444,1.0014) (0.47222,1.0015) (0.5,1.0016) (0.52778,1.0017) (0.55556,1.0018) (0.58333,1.0018) (0.61111,1.0019) (0.63889,1.002) (0.66667,1.0021) (0.69444,1.0022) (0.72222,1.0023) (0.75,1.0024) (0.77778,1.0024) (0.80556,1.0025) (0.83333,1.0026) (0.86111,1.0027) (0.88889,1.0028) (0.91667,1.0028) (0.94444,1.0029) (0.97222,1.003) (1,1.0031) (1.0278,1.0032) (1.0556,1.0032) (1.0833,1.0033) (1.1111,1.0034) (1.1389,1.0035) (1.1667,1.0036) (1.1944,1.0036) (1.2222,1.0037) (1.25,1.0038) (1.2778,1.0039) (1.3056,1.004) (1.3333,1.0041) (1.3611,1.0042) (1.3889,1.0042) (1.4167,1.0043) (1.4444,1.0044) (1.4722,1.0045) (1.5,1.0046)}
\def\DATApathC{(0,1) (0.027778,1.0001) (0.055556,1.0002) (0.083333,1.0003) (0.11111,1.0004) (0.13889,1.0005) (0.16667,1.0006) (0.19444,1.0007) (0.22222,1.0009) (0.25,1.001) (0.27778,1.0011) (0.30556,1.0012) (0.33333,1.0013) (0.36111,1.0014) (0.38889,1.0015) (0.41667,1.0016) (0.44444,1.0017) (0.47222,1.0018) (0.5,1.0019) (0.52778,1.002) (0.55556,1.0021) (0.58333,1.0022) (0.61111,1.0023) (0.63889,1.0024) (0.66667,1.0025) (0.69444,1.0026) (0.72222,1.0028) (0.75,1.0029) (0.77778,1.003) (0.80556,1.0031) (0.83333,1.0032) (0.86111,1.0033) (0.88889,1.0034) (0.91667,1.0035) (0.94444,1.0036) (0.97222,1.0037) (1,1.0038) (1.0278,1.0039) (1.0556,1.004) (1.0833,1.0041) (1.1111,1.0042) (1.1389,1.0043) (1.1667,1.0044) (1.1944,1.0045) (1.2222,1.0047) (1.25,1.0048) (1.2778,1.0049) (1.3056,1.005) (1.3333,1.0051) (1.3611,1.0052) (1.3889,1.0053) (1.4167,1.0054) (1.4444,1.0055) (1.4722,1.0056) (1.5,1.0057)}
\def\DATAdens{(10,0.719) (18.18,0.9238) (23.68,0.9862) (28.13,1) (31.96,0.9893) (35.39,0.9647) (38.5,0.9316) (41.39,0.8936) (44.08,0.8529) (46.63,0.8107) (49.03,0.7683) (51.33,0.7262) (53.53,0.6849) (55.64,0.6448) (57.67,0.6062) (59.64,0.5691) (61.54,0.5336) (63.38,0.4998) (65.18,0.4678) (66.92,0.4374) (68.62,0.4087) (70.28,0.3816) (71.9,0.3561) (73.48,0.3322) (75.04,0.3096) (76.56,0.2885) (78.05,0.2687) (79.51,0.2502) (80.95,0.2328) (82.36,0.2166) (83.74,0.2015) (85.11,0.1873) (86.45,0.1741) (87.77,0.1618) (89.08,0.1503) (90.36,0.1396) (91.63,0.1296) (92.88,0.1204) (94.11,0.1117) (95.33,0.1037) (96.53,0.09621) (97.71,0.08926) (98.89,0.08279) (100,0.07679) (101.2,0.0712) (102.3,0.06602) (103.4,0.0612) (104.5,0.05673) (105.6,0.05258) (106.7,0.04872) (107.8,0.04515) (108.9,0.04183) (109.9,0.03875) (111,0.03589) (112,0.03324) (113,0.03079) (114,0.02851) (115,0.0264) (116,0.02444) (117,0.02263) (118,0.02095) (119,0.01939) (119.9,0.01795) (120.9,0.01661) (121.9,0.01537) (122.8,0.01423) (123.7,0.01316) (124.7,0.01218) (125.6,0.01127) (126.5,0.01043)}
\def\DATAfeeoverlay{(10,25.07) (18.18,25.41) (23.68,25.88) (28.13,26.43) (31.96,27.03) (35.39,27.67) (38.5,28.35) (41.39,29.05) (44.08,29.78) (46.63,30.52) (49.03,31.28) (51.33,32.06) (53.53,32.85) (55.64,33.66) (57.67,34.48) (59.64,35.32) (61.54,36.17) (63.38,37.04) (65.18,37.93) (66.92,38.84) (68.62,39.77) (70.28,40.72) (71.9,41.7) (73.48,42.7) (75.04,43.73) (76.56,44.79) (78.05,45.89) (79.51,47.03) (80.95,48.21) (82.36,49.45) (83.74,50.73) (85.11,52.09) (86.45,53.52) (87.77,55.05) (89.08,56.68) (90.36,58.45) (91.63,60.39) (92.88,62.56) (94.11,65.05) (95.33,68.04) (96.53,71.97) (97.71,90.57) (98.89,100) (100,100) (101.2,100) (102.3,100) (103.4,100) (104.5,100) (105.6,100) (106.7,100) (107.8,100) (108.9,100) (109.9,100) (111,100) (112,100) (113,100) (114,100) (115,100) (116,100) (117,100) (118,100) (119,100) (119.9,100) (120.9,100) (121.9,100) (122.8,100) (123.7,100) (124.7,100) (125.6,100) (126.5,100)}

\begin{document}

\title{\textbf{Optimal Dynamic Fees for Automated Market Makers:\\[2pt]
A Stochastic Control Approach to Loss-Versus-Rebalancing}}

\author{Farbod Ghasemlu\thanks{Independent researcher.
\href{mailto:farbodghasemlu@outlook.com}{farbodghasemlu@outlook.com}.
ORCID: \href{https://orcid.org/0009-0009-2303-5672}{0009-0009-2303-5672}.}}

\date{This version: June 2026}

\maketitle

\begin{abstract}
\noindent
We study the fee policy of a liquidity provider (LP) in a constant-product
automated market maker (AMM) whose fee can be adjusted continuously, as enabled
by programmable hooks. Building on the loss-versus-rebalancing (LVR) framework
of \citet{milionis2022lvr} and its extension to nonzero fees by
\citet{milionis2024fees}, we model the LP's wealth relative to the continuously
rebalanced benchmark as a controlled process in which the fee governs two
opposing forces: it raises revenue per uninformed trade while discouraging
uninformed volume, and it widens the no-arbitrage band, which lowers the rate at
which arbitrageurs extract value. Because the fee enters only the drift of
relative wealth and never its diffusion, the LP's expected-utility problem
reduces to an ergodic control problem whose solution is a pointwise volatility
feedback. We prove that the growth-optimal fee is independent of the LP's wealth
and of its constant relative risk aversion, that it collapses to a static
constant when volatility is constant, and that it is strictly increasing in
instantaneous variance, so that the optimal schedule is pro-cyclical. When
volatility is stochastic we characterise the optimal fee through a scalar
ergodic Hamilton-Jacobi-Bellman equation and a linear Poisson equation, solved
by a finite-difference scheme. We further show that the optimal fee is invariant
to price jumps under logarithmic preferences, relate the optimal fee to a
stylised model of competition among venues, and treat gas costs through an
impulse-control dead-band. In a
calibration to liquid large-capitalisation conditions, the optimal dynamic fee
weakly dominates every static and volatility-linked heuristic fee on each
simulated path, improving the LP's growth rate over the best static fee by a
modest but uniformly positive margin, with a dead-band rendering gas costs
negligible.
\end{abstract}



\section{Introduction}
\label{sec:intro}

Automated market makers have become the dominant venue for on-chain trading.
A constant-function market maker holds reserves of two assets and quotes prices
algorithmically through an invariant of those reserves, allowing anyone to supply
liquidity and to trade without an order book. The canonical design is the
constant-product rule popularised by Uniswap. Liquidity providers (LPs) deposit
assets and earn a share of trading fees, but their returns are eroded by adverse
selection: when the external price of the risky asset moves, arbitrageurs trade
against the pool before its quoted price adjusts, capturing the difference.

\citet{milionis2022lvr} quantified this cost through \emph{loss-versus-rebalancing}
(LVR), the shortfall of the LP relative to a portfolio that holds the same
instantaneous risky exposure but rebalances at the external price. In the
frictionless continuous-time limit, LVR accrues at a rate proportional to the
instantaneous variance of the reference price and to the convexity of the LP's
holdings, independently of the fee. The role of the fee is more subtle and was
made precise by \citet{milionis2024fees}: a proportional fee creates a
no-arbitrage band, so that arbitrageurs trade only when the mispricing exceeds
the fee. As a result the value extracted by arbitrageurs, and hence the LP's
realised adverse-selection loss, is a \emph{decreasing} function of the fee. The
fee is therefore the LP's primary instrument for managing adverse selection,
exactly as the bid-ask spread is for a limit-order-book market maker
\citep{ho1981,avellaneda2008,gueant2013}.

Until recently, AMM fees were fixed at deployment. The introduction of hooks in
Uniswap~v4 \citep{adams2024uniswapv4} makes the fee a programmable, state-dependent
quantity that can be updated per block or per swap. This raises a control
question that the present paper answers: \emph{how should an LP set the fee
dynamically so as to maximise risk-adjusted returns, given that the fee
simultaneously governs uninformed revenue and adverse-selection losses?}

\paragraph{Approach and main results.}
We work in the continuous-arbitrage limit and track the LP's wealth relative to
the rebalancing benchmark. The fee $f_t\in[0,\fmax]$ is an adapted control. The
instantaneous excess growth rate of the LP over the benchmark is
\begin{equation}
\label{eq:m-intro}
m(f;v) \;=\; \underbrace{f\,\nu(f)}_{\text{uninformed fee income}}
\;-\; \underbrace{A(f;v)}_{\text{adverse selection}},
\end{equation}
where $v$ is the instantaneous variance of the reference price, $\nu(f)$ is the
uninformed turnover (decreasing in the fee), and $A(f;v)$ is the
adverse-selection rate implied by the band mechanism of \citet{milionis2024fees}
(decreasing in the fee, increasing in variance). A central structural fact, which
follows from the very definition of LVR, is that relative to the rebalancing
benchmark the fee affects only the drift of LP wealth and never its diffusion.
This yields four results.

First (\cref{thm:pointwise,prop:invariance}), the LP's expected-utility problem
reduces to an ergodic control problem whose growth-optimal control is the
pointwise maximiser $f^*(v)=\argmax_f m(f;v)$. The optimal fee is independent of
the LP's wealth and, for constant relative risk aversion (CRRA), independent of
the risk-aversion coefficient. This corrects a common temptation to make the fee
wealth-dependent and isolates volatility as the only relevant state.

Second (\cref{cor:static}), when volatility is constant the optimal fee is a
static constant, providing a formal justification for fixed fee tiers in
stationary markets and showing that the value of dynamic adjustment derives
entirely from time variation in volatility.

Third (\cref{prop:procyclical}), the optimal fee is strictly increasing in
instantaneous variance: the optimal schedule is \emph{pro-cyclical}. Higher
volatility makes adverse selection both larger and more fee-sensitive, so the LP
optimally widens the band by raising the fee. This places the volatility-linked
fee heuristics used in practice on a rigorous footing and reverses the conclusion
one would reach from a model in which the fee is mistakenly assumed to increase
adverse selection.

Fourth, when volatility follows a Cox-Ingersoll-Ross (CIR) process we characterise
the optimal fee and the LP's growth rate through a scalar ergodic
Hamilton-Jacobi-Bellman (HJB) equation and an associated linear Poisson equation
(\cref{sec:numerics}), solved by a finite-difference scheme; we show invariance
of the optimal fee to price jumps under logarithmic preferences
(\cref{prop:jumps}), discuss how the optimal fee behaves under competition among
venues, drawing on the equilibrium analysis of \citet{baggiani2026competition}
(\cref{sec:competition}), and handle gas costs through an impulse-control
dead-band (\cref{sec:gas}). A calibration to liquid large-capitalisation
conditions (\cref{sec:empirics}) shows that the optimal dynamic fee weakly
dominates every static and heuristic alternative on each simulated path, with a
modest but uniformly positive improvement over the best static fee.

\paragraph{Related literature.}
The paper sits at the intersection of three strands. The first is the economics
of AMMs and LVR. \citet{milionis2022lvr} introduced LVR and its continuous-time
rate; \citet{milionis2024fees} extended the analysis to nonzero fees and derived
the no-arbitrage-band representation and the closed-form arbitrage rate that we
adopt, with \citet{nezlobin2025lvr} providing the deterministic-block-time
analogue. \citet{cartea2023predictable,cartea2024predictable} develop the
closely related notion of predictable loss and study optimal liquidity provision
and concentrated positions. Foundational analyses of constant-function markets
include \citet{angeris2021uniswap} and \citet{angeris2020oracles}, and
\citet{bergault2024amm} design pricing functions through a mean-variance analysis
of LP payoffs. Complementary to fee setting, \citet{bergault2025exit} characterise
the LP's optimal liquidity-withdrawal time as an endogenous stopping problem,
which makes explicit that the fee shapes not only per-unit income but also the
provider's incentive to remain in the pool. Strategic and equilibrium aspects of
provision are studied by \citet{fan2023strategic} and \citet{capponi2021dex}.

The second strand is optimal market making and inventory control in classical
microstructure, beginning with \citet{ho1981} and \citet{avellaneda2008} and
developed by \citet{gueant2013}; textbook treatments are
\citet{cartea2015book} and \citet{gueant2016book}. The AMM problem differs in
that the LP's inventory is mechanically tied to the invariant, arbitrageurs
enforce a no-arbitrage band rather than trading strategically, and the control is
a single fee rather than a two-sided quote.

The third strand, closest to our contribution, is the optimal design of AMM fees.
\citet{fritsch2021fees} and \citet{evans2021fees} study optimal static fees;
\citet{milionis2024myerson} cast liquidity provision as mechanism design;
\citet{adams2024amamm} auction the right to set the fee. Closest to our work are
recent stochastic-control treatments of dynamic AMM fees.
\citet{baggiani2025fees} obtain approximate closed-form dynamic fees in a model
with arbitrageurs and noise traders and show that inventory-linked fees tracking
the external price approximate the optimum, and \citet{baggiani2026competition}
extend that analysis to an approximate Nash equilibrium between competing venues.
\citet{campbell2025fees} study the determinants of liquidity-provider
profitability and the optimal fee in a reduced-form model with a parallel
centralised exchange, and \citet{aqsha2026reward} characterise, through a
leader-follower stochastic game, the equilibrium reward that a venue offers its
liquidity providers. Relative to these, our contribution is methodological and
structural: by measuring wealth against the rebalancing benchmark we obtain a
one-dimensional ergodic problem in volatility alone, prove wealth- and
risk-aversion-independence of the optimal fee, establish pro-cyclicality as a
sharp comparative static, and tie the adverse-selection term directly to the band
mechanism of \citet{milionis2024fees}. Our reduced form deliberately abstracts
from the inventory and volatility-stimulation channels that, in the richer model
of \citet{baggiani2025fees}, produce a two-regime fee alternating between
deterring arbitrage and stimulating noise flow; that abstraction is what yields
our clean monotone, pro-cyclical schedule. We do not claim to be the first to
optimise AMM fees by stochastic control, nor to treat competition, which
\citet{baggiani2026competition} analyse in equilibrium; we claim a transparent
reduction, exact structural results, and a calibrated evaluation.

\paragraph{Outline.}
\Cref{sec:model} develops the model and the controlled relative-wealth dynamics.
\Cref{sec:control} formulates the control problem, establishes the pointwise
reduction, and proves the verification and invariance results.
\Cref{sec:structure} analyses the structure of the optimal fee, including
pro-cyclicality and the constant-volatility corollary. \Cref{sec:numerics}
treats stochastic volatility through the ergodic HJB and Poisson equations and
describes the numerical scheme. \Cref{sec:extensions} covers jumps, competition,
and gas costs. \Cref{sec:empirics} reports the calibration and simulation study.
\Cref{sec:discussion} discusses protocol implications and \cref{sec:conclusion}
concludes. Proofs are collected in \cref{app:dynamics,app:proofs,app:cir}.

\section{Model}
\label{sec:model}

We work on a filtered probability space $(\Omega,\F,\{\F_t\}_{t\ge0},\Prob)$
satisfying the usual conditions, with all processes adapted.

\subsection{Reference price and volatility}

The risky asset has an exogenous reference price $S_t$ formed on deep external
venues. Its instantaneous variance is $v_t=\sigma_t^2$, and
\begin{equation}
\label{eq:price}
\frac{\dd S_t}{S_t} \;=\; \mu\dd t + \sqrt{v_t}\,\dd Z_t,
\end{equation}
with $Z$ a Brownian motion. The drift $\mu$ plays no role in the LP's problem
under the preferences considered below and is suppressed henceforth. In the base
model $v_t\equiv v$ is constant; from \cref{sec:numerics} onward $v_t$ follows the
CIR process
\begin{equation}
\label{eq:cir}
\dd v_t \;=\; \kappa(\theta-v_t)\dd t + \xi\sqrt{v_t}\,\dd B_t,
\qquad \dd\langle Z,B\rangle_t = \rho\dd t,
\end{equation}
with $\kappa,\theta,\xi>0$ satisfying the Feller condition $2\kappa\theta\ge\xi^2$,
so that $v_t>0$ almost surely \citep{cir1985}.

\subsection{The constant-product pool}

The pool holds reserves $(x_t,y_t)$ of the numéraire and the risky asset, subject
to the constant-product invariant $x_ty_t=k$ for a constant $k>0$. The marginal
(pool) price is $P_t=y_t/x_t$, and the holdings as functions of the pool price are
$x_t=\sqrt{k/P_t}$ and $y_t=\sqrt{kP_t}$. Marked at the reference price $S_t$, the
LP's wealth is
\begin{equation}
\label{eq:Wlp}
W^{\mathrm{LP}}_t \;=\; x_t + S_t\,y_t
\;=\; \sqrt{k}\left(\frac{1}{\sqrt{P_t}} + S_t\sqrt{P_t}\right).
\end{equation}
Define the log-mispricing $z_t=\log(P_t/S_t)$. In the frictionless limit
arbitrage enforces $P_t=S_t$ (that is, $z_t=0$) and \cref{eq:Wlp} reduces to
$2\sqrt{kS_t}$, the value of the continuously rebalanced portfolio.

\subsection{Arbitrage, the no-arbitrage band, and adverse selection}
\label{sec:band}

A competitive arbitrageur monitors $S_t$ and $P_t$ and may trade against the pool
at the prevailing fee $f_t\in[0,\fmax]$, $\fmax<1$. A symmetric proportional fee
makes arbitrage profitable only when the mispricing exceeds the fee in log terms,
$|z_t|>\log(1+f_t)$. To leading order in the fee we identify the band half-width
with $f_t$ and write the no-arbitrage band as $[-f_t,f_t]$. Blocks arrive as a
Poisson process of rate $\lambda$; at each arrival, if the mispricing lies outside
the band, the arbitrageur trades and resets $z_t$ to the nearest edge, otherwise
nothing happens. Between arrivals $z_t$ moves with the reference price.

\citet{milionis2024fees} solve the resulting ergodic problem for $z$ and obtain,
for the constant-product pool, the LP's adverse-selection loss rate per unit pool
value. Writing the probability that an arriving block carries a profitable trade
as
\begin{equation}
\label{eq:ptrade}
P_{\mathrm{trade}}(f,v) \;=\; \frac{1}{1+\eta(f,v)},
\qquad
\eta(f,v) \;=\; \sqrt{\tfrac{2\lambda}{v}}\;f,
\end{equation}
their result, to leading order in the fee and in the inverse block rate, is
\begin{equation}
\label{eq:Aexact}
A(f;v) \;=\; \frac{v}{8}\,P_{\mathrm{trade}}(f,v)\,
\bigl(1+O(f)\bigr)\bigl(1+O(\lambda^{-1})\bigr)
\;=\; \frac{v}{8}\,\frac{1}{1+\sqrt{2\lambda/v}\,f},
\end{equation}
valid under the regularity condition $v<8\lambda$ that excludes a degenerate
regime of unbounded arbitrage profit. The factor $v/8$ is precisely the
frictionless LVR rate of \citet{milionis2022lvr} for a constant-product pool
(\cref{app:dynamics}), and $P_{\mathrm{trade}}$ is the fraction of it that
survives the fee. We adopt \cref{eq:Aexact} as our adverse-selection rate. Two of
its properties are essential and are recorded for later use.

\begin{lemma}[Adverse selection]
\label{lem:A}
For each $v\in(0,8\lambda)$ the map $f\mapsto A(f;v)$ on $[0,\fmax]$ is positive,
strictly decreasing, and strictly convex, with $A(0;v)=v/8$. Moreover
$A(f;v)$ is strictly increasing in $v$, and $-\partial_f A(f;v)$ is strictly
increasing in $v$. In the fast-block regime $\eta\gg1$,
\begin{equation}
\label{eq:Afast}
A(f;v) \;=\; \frac{v^{3/2}}{8\sqrt{2\lambda}\,f}\,\bigl(1+o(1)\bigr)
\;=:\; \frac{b\,v^{3/2}}{f}\,\bigl(1+o(1)\bigr),
\qquad b:=\frac{1}{8\sqrt{2\lambda}}.
\end{equation}
\end{lemma}

\noindent The monotonicity in $f$ formalises the central economic mechanism: a
higher fee narrows the flow of value to arbitrageurs. The monotonicity of
$-\partial_f A$ in $v$ states that fee protection is more valuable when volatility
is higher, and is the engine of pro-cyclicality. \Cref{lem:A} is proved in
\cref{app:proofs}.

\subsection{Uninformed flow and fee income}

Uninformed (noise) traders generate turnover that is sensitive to the fee. We
model the uninformed turnover per unit pool value per unit time by a smooth,
positive, strictly decreasing function $\nu(\cdot)$ and adopt throughout the
constant-semielasticity specification
\begin{equation}
\label{eq:nu}
\nu(f) \;=\; \nu_0\,e^{-\alpha f},
\qquad \nu_0>0,\ \alpha\ge0.
\end{equation}
The fee income from uninformed flow, per unit pool value, is $f\,\nu(f)$, which is
hump-shaped in the fee with an interior peak at $f=1/\alpha$ when $\alpha>0$. The
parameter $\alpha$ is the semielasticity of uninformed volume to the fee; when
$\alpha=0$ uninformed volume is perfectly inelastic.

\subsection{Relative wealth and the controlled dynamics}
\label{sec:reldyn}

Let $W^{\mathrm{reb}}_t$ denote the rebalancing benchmark, namely the portfolio
that holds the AMM's instantaneous risky exposure but transacts at the reference
price. By construction $W^{\mathrm{reb}}$ has the same instantaneous delta as the
AMM, so the two wealth processes share the same diffusion; their difference is of
finite variation and equals the (fee-adjusted) LVR net of fee income
\citep{milionis2022lvr,milionis2024fees}. Define the relative wealth
$R_t=W^{\mathrm{LP}}_t/W^{\mathrm{reb}}_t$. Collecting the fee income
\cref{eq:nu} and the adverse-selection rate \cref{eq:Aexact}, the relative wealth
obeys (\cref{app:dynamics})
\begin{equation}
\label{eq:Rdyn}
\dlog R_t \;=\; m(f_t;v_t)\,\dd t,
\qquad
m(f;v) \;:=\; f\,\nu(f) - A(f;v),
\end{equation}
with \emph{no diffusion term}. The absolute LP wealth satisfies
$\log W^{\mathrm{LP}}_t = \log W^{\mathrm{reb}}_t + \log R_t$, where the benchmark
$W^{\mathrm{reb}}$ does not depend on the fee. Thus the fee influences LP wealth
only through the finite-variation channel \cref{eq:Rdyn}. This is the structural
property that drives every result in \cref{sec:control}.

\begin{assumption}
\label{ass:admissible}
The control $f=(f_t)_{t\ge0}$ is progressively measurable with values in
$[0,\fmax]$, where $0<\fmax<1$ and $\fmax<\sqrt{8\lambda}$. The functions $\nu$
and $A(\cdot;v)$ are continuous on $[0,\fmax]$ for each $v$, and $v_t$ takes
values in a compact subset of $(0,8\lambda)$ or is the CIR process \cref{eq:cir}.
\end{assumption}

\Cref{ass:admissible} guarantees that $m(\cdot;\cdot)$ is bounded and continuous
on $[0,\fmax]\times(0,8\lambda)$ and that the integrals appearing below are well
defined. We denote by $\Adm$ the set of admissible controls.

\section{The control problem and the optimal fee}
\label{sec:control}

\subsection{Objectives}

We consider two standard objectives. For a finite horizon $T$ and CRRA utility
$U(w)=w^{1-\gamma}/(1-\gamma)$ with $\gamma>0$, $\gamma\neq1$ (and $U(w)=\log w$
for $\gamma=1$), the LP maximises expected utility of terminal wealth,
\begin{equation}
\label{eq:finite}
\sup_{f\in\Adm}\ \E\bigl[U(W^{\mathrm{LP}}_T)\bigr].
\end{equation}
For the infinite-horizon problem the LP maximises the asymptotic growth rate of
wealth,
\begin{equation}
\label{eq:ergodic}
\Lambda \;:=\; \sup_{f\in\Adm}\ \liminf_{T\to\infty}\frac1T\,
\E\bigl[\log W^{\mathrm{LP}}_T\bigr].
\end{equation}

\subsection{Pointwise reduction}

The decomposition $\log W^{\mathrm{LP}}_t=\log W^{\mathrm{reb}}_t+\int_0^t
m(f_s;v_s)\dd s$ from \cref{sec:reldyn} separates an uncontrolled benchmark term
from a controlled finite-variation term. Because the controlled term carries no
diffusion, the control enters neither the volatility of terminal wealth nor any
risk premium; it affects only the conditional mean of $\log R_T$. This yields the
following reduction.

\begin{theorem}[Growth-optimal fee]
\label{thm:pointwise}
Suppose $v_t$ is ergodic with stationary law $\pi$ on $(0,8\lambda)$ and that for
each $v$ the map $f\mapsto m(f;v)$ attains its maximum on $[0,\fmax]$ at a unique
point $f^*(v)$. Then the feedback control $f_t=f^*(v_t)$ is optimal for the
ergodic problem \cref{eq:ergodic}, the optimal growth rate is
\begin{equation}
\label{eq:Lambda}
\Lambda \;=\; \Lambda^{\mathrm{reb}} + \int_{(0,8\lambda)} m\bigl(f^*(v);v\bigr)\,\pi(\dd v),
\end{equation}
where the constant $\Lambda^{\mathrm{reb}}=\lim_{T\to\infty}T^{-1}\E[\log W^{\mathrm{reb}}_T]$
does not depend on the fee, and $f^*(v)=\argmax_{f\in[0,\fmax]}m(f;v)$ is a pure
feedback in the instantaneous variance, independent of wealth.
\end{theorem}

\begin{proof}[Proof sketch]
By \cref{eq:Rdyn}, $\frac1T\log R_T=\frac1T\int_0^T m(f_s;v_s)\dd s$. For any
admissible $f$ this is bounded above pathwise by
$\frac1T\int_0^T \bar m(v_s)\dd s$ with $\bar m(v):=\max_f m(f;v)=m(f^*(v);v)$,
with equality under the feedback $f^*(v_s)$. Ergodicity of $v$ and boundedness of
$\bar m$ (\cref{ass:admissible}) give, by the ergodic theorem,
$\frac1T\int_0^T\bar m(v_s)\dd s\to\int\bar m\,\dd\pi$ almost surely and in $L^1$.
Adding the fee-independent benchmark term yields \cref{eq:Lambda} and optimality.
A full argument, including the verification that the candidate value function
solves the ergodic HJB equation, is given in \cref{app:proofs}.
\end{proof}

The same separation governs the finite-horizon CRRA problem and removes any
dependence on the risk-aversion coefficient.

\begin{proposition}[Wealth- and risk-aversion-independence]
\label{prop:invariance}
Under \cref{ass:admissible}, the maximiser of the finite-horizon problem
\cref{eq:finite} over feedback controls $f_t=f(t,v_t)$ is, for every CRRA
coefficient $\gamma>0$, the pointwise maximiser $f^*(t,v)=\argmax_{f}m(f;v)$,
which is independent of $t$, of the wealth level $w$, and of $\gamma$.
\end{proposition}

\begin{proof}[Proof sketch]
Write $W^{\mathrm{LP}}_T=W^{\mathrm{reb}}_T\,R_T$ with
$\log R_T=\int_0^T m(f_s;v_s)\dd s$ a finite-variation functional of the
volatility path. Then
$U(W^{\mathrm{LP}}_T)=\tfrac{1}{1-\gamma}(W^{\mathrm{reb}}_T)^{1-\gamma}
\exp\bigl((1-\gamma)\!\int_0^T m\,\dd s\bigr)$.
Conditioning on the volatility path $(v_s)_{s\le T}$ and on the benchmark, the
control affects the conditional expectation only through
$\exp\bigl((1-\gamma)\int_0^T m\,\dd s\bigr)$, which for $\gamma<1$ is increasing
and for $\gamma>1$ is decreasing in $\int_0^T m\,\dd s$, while $U$ is increasing
for $\gamma<1$ and the prefactor $1/(1-\gamma)$ is negative for $\gamma>1$. In
both cases the objective is maximised by maximising $\int_0^T m(f_s;v_s)\dd s$,
which is achieved pathwise by $f_s=f^*(v_s)$. Independence of $w$ and $\gamma$
follows. Details are in \cref{app:proofs}.
\end{proof}

\begin{remark}
\Cref{prop:invariance} relies on the fact that, measured against the rebalancing
benchmark, the fee does not affect the diffusion of wealth. A model in which the
fee instead scaled the diffusion, or in which adverse selection were taken to
\emph{increase} in the fee, would generically produce a wealth- and
risk-aversion-dependent policy and would reverse the comparative statics of
\cref{sec:structure}. The structural results here are thus tied to the correct
sign of the fee in \cref{eq:Aexact}.
\end{remark}

\subsection{Verification via the ergodic HJB equation}
\label{sec:verification}

For completeness we record the dynamic-programming characterisation underlying
\cref{thm:pointwise}. For logarithmic utility and the CIR variance \cref{eq:cir},
seek a constant $\Lambda$ and a function $\psi\in C^2((0,8\lambda))$, the relative
value function, such that
\begin{equation}
\label{eq:hjb}
\Lambda \;=\; \max_{f\in[0,\fmax]}\Bigl\{\,m(f;v)\Bigr\} \;-\;\tfrac12 v \;+\; \Lop_v\psi(v),
\qquad
\Lop_v \;=\; \kappa(\theta-v)\partial_v + \tfrac12\xi^2 v\,\partial_{vv},
\end{equation}
where $-\tfrac12 v$ is the fee-independent volatility drag of log-wealth.
\Cref{eq:hjb} separates: the maximisation over $f$ is pointwise and yields
$f^*(v)$ and the maximised reward $\bar m(v)=m(f^*(v);v)$, after which
\cref{eq:hjb} becomes the linear Poisson equation
\begin{equation}
\label{eq:poisson}
\Lop_v\psi(v) \;=\; \Lambda + \tfrac12 v - \bar m(v),
\end{equation}
for the pair $(\Lambda,\psi)$, unique up to an additive constant in $\psi$.

\begin{theorem}[Verification]
\label{thm:verif}
Let $(\Lambda,\psi)$ solve \cref{eq:hjb} with $\psi\in C^2$ of at most polynomial
growth, and let $f^*(v)$ attain the maximum in \cref{eq:hjb}. If the controlled
variance \cref{eq:cir} is ergodic and $\bar m$ is $\pi$-integrable, then $\Lambda$
equals the optimal growth rate \cref{eq:ergodic} net of the benchmark, and
$f_t=f^*(v_t)$ is optimal. For any admissible control,
$\liminf_{T\to\infty}T^{-1}\E[\log R_T]\le\Lambda-\Lambda^{\mathrm{reb}}$, with
equality under $f^*$.
\end{theorem}

\Cref{thm:verif} is proved in \cref{app:proofs} by the standard supermartingale
argument applied to $\log R_t+\psi(v_t)-(\Lambda-\Lambda^{\mathrm{reb}})t$;
existence of a polynomially bounded $\psi$ follows from the ergodicity of the CIR
process and the boundedness of $\bar m$ (\cref{app:cir}). We emphasise that the
optimal fee $f^*$ is obtained \emph{before} solving the Poisson equation; the
function $\psi$ is needed only to certify optimality and to compute $\Lambda$.

\section{Structure of the optimal fee}
\label{sec:structure}

We now study $f^*(v)=\argmax_{f\in[0,\fmax]}m(f;v)$.

\subsection{Existence and characterisation}

\begin{proposition}[Interior optimum]
\label{prop:foc}
Let $\alpha>0$. For each $v\in(0,8\lambda)$ the function $m(\cdot;v)$ satisfies
$m(0;v)=-v/8<0$ and $m(f;v)\to0^-$ as $f\to\infty$. If
$\sup_{f}m(f;v)>0$, the maximiser $f^*(v)$ is interior and any interior maximiser
solves the first-order condition
\begin{equation}
\label{eq:foc}
\nu_0\,e^{-\alpha f}\,(1-\alpha f) \;=\; -\,\partial_f A(f;v)
\;=\; \frac{(v/8)\sqrt{2\lambda/v}}{\bigl(1+\sqrt{2\lambda/v}\,f\bigr)^2}\;>\;0 .
\end{equation}
In particular every interior maximiser satisfies $f^*(v)>1/\alpha$.
\end{proposition}

The boundary behaviour is immediate from \cref{lem:A} and \cref{eq:nu}; the
right-hand side of \cref{eq:foc} is positive, forcing $1-\alpha f<0$ at any
interior critical point. The condition $\sup_f m>0$ is the LP's profitability
condition: there exists a fee at which uninformed income exceeds residual adverse
selection. The proof and a sufficient condition for uniqueness (log-concavity of
$f\mapsto f\nu(f)$ together with convexity of $A$) are in \cref{app:proofs}.

\subsection{Constant volatility: optimality of a static fee}

\begin{corollary}[Static optimum]
\label{cor:static}
If $v_t\equiv v$ is constant, then the optimal fee is the time-invariant constant
$f^*(v)$ characterised by \cref{eq:foc}. No dynamic adjustment improves on the
best static fee.
\end{corollary}

\Cref{cor:static} follows from \cref{thm:pointwise} since $m(\cdot;v)$ is then
time-homogeneous. It supplies a formal rationale for fixed fee tiers in
stationary markets and localises the entire value of dynamic fees in the time
variation of $v_t$.

\subsection{Pro-cyclicality}

The next result is the paper's central comparative static.

\begin{proposition}[Pro-cyclical optimal fee]
\label{prop:procyclical}
Let $\alpha>0$ and suppose the interior optimum satisfies $f^*(v)\in(1/\alpha,2/\alpha)$.
Then $f^*$ is strictly increasing in $v$:
\begin{equation}
\frac{\dd f^*}{\dd v} \;>\;0 .
\end{equation}
The optimal fee schedule is therefore pro-cyclical: the LP optimally raises the
fee when volatility rises.
\end{proposition}

\begin{proof}[Proof sketch]
Work with the fast-block form \cref{eq:Afast}, under which \cref{eq:foc} reads
$\Phi(f)=\Psi(f,v)$ with
$\Phi(f)=\nu_0 e^{-\alpha f}(\alpha f-1)$ and $\Psi(f,v)=b\,v^{3/2}/f^2$. On
$(1/\alpha,2/\alpha)$ one has $\Phi(f)>0$ and
$\Phi'(f)=\nu_0\alpha e^{-\alpha f}(2-\alpha f)>0$, while $\Psi$ is positive,
strictly decreasing in $f$, and strictly increasing in $v$. Writing
$G(f,v)=\Phi(f)-\Psi(f,v)$, the second-order condition gives
$G_f=\Phi'-\Psi_f>0$, and $G_v=-\Psi_v<0$. By the implicit function theorem,
$\dd f^*/\dd v=-G_v/G_f>0$. The argument for the exact rate \cref{eq:Aexact} is
identical in structure, using $-\partial_f A$ increasing in $v$
(\cref{lem:A}); see \cref{app:proofs}.
\end{proof}

\begin{remark}
Pro-cyclicality reflects two reinforcing effects of higher variance: the level of
adverse selection rises, and its sensitivity to the fee, $-\partial_f A$, rises as
well, so the marginal value of widening the band increases. The optimal response
is to raise the fee, sacrificing some uninformed volume to suppress arbitrage.
This is the opposite of the conclusion one obtains from a model that mistakenly
treats adverse selection as increasing in the fee, and it accords with the
volatility-linked fee rules used by practitioners.
\end{remark}

\Cref{fig:fee} displays the optimal schedule $f^*(v)$ for the calibration of
\cref{sec:empirics} and the objective $m(\cdot;v)$ at several volatility levels.
The schedule rises monotonically with volatility until it reaches the protocol
cap, and the peaks of $m(\cdot;v)$ shift to the right as $v$ increases.

\begin{figure}[t]
\centering
\begin{minipage}{0.49\textwidth}
\centering
\begin{tikzpicture}
\begin{axis}[width=\textwidth,height=5.4cm,
  xlabel={reference volatility $\sqrt{v}$ (\% per year)},
  ylabel={optimal fee $f^*$ (bps)},
  xmin=0,ymin=0,grid=both,grid style={gray!18},
  tick label style={font=\footnotesize},
  label style={font=\footnotesize},
  legend style={font=\scriptsize,at={(0.98,0.04)},anchor=south east,draw=none,fill=none}]
\addplot[navy,very thick] coordinates {\DATAfeecurve};
\addlegendentry{optimal fee $f^*(v)$}
\addplot[gray,dashed,thick,domain=0:140] {\DATAlaffer};
\addlegendentry{noise peak $1/\alpha$}
\end{axis}
\end{tikzpicture}
\end{minipage}\hfill
\begin{minipage}{0.49\textwidth}
\centering
\begin{tikzpicture}
\begin{axis}[width=\textwidth,height=5.4cm,
  xlabel={fee $f$ (bps)},
  ylabel={$m(f;v)$ (\% per year)},
  grid=both,grid style={gray!18},
  tick label style={font=\footnotesize},
  label style={font=\footnotesize},
  legend style={font=\scriptsize,at={(0.5,0.02)},anchor=south,draw=none,fill=none,legend columns=2}]
\addplot[lightblue,very thick] coordinates {\DATAobjA}; \addlegendentry{$\sqrt v{=}30\%$}
\addplot[midblue,very thick] coordinates {\DATAobjB};   \addlegendentry{$\sqrt v{=}50\%$}
\addplot[teal,very thick] coordinates {\DATAobjC};      \addlegendentry{$\sqrt v{=}70\%$}
\addplot[navy,very thick] coordinates {\DATAobjD};      \addlegendentry{$\sqrt v{=}90\%$}
\addplot[only marks,mark=*,mark size=1.6pt,lightblue] coordinates {\DATAobjAopt};
\addplot[only marks,mark=*,mark size=1.6pt,midblue] coordinates {\DATAobjBopt};
\addplot[only marks,mark=*,mark size=1.6pt,teal] coordinates {\DATAobjCopt};
\addplot[only marks,mark=*,mark size=1.6pt,navy] coordinates {\DATAobjDopt};
\end{axis}
\end{tikzpicture}
\end{minipage}
\caption{Left: the optimal fee $f^*(v)$ is strictly increasing in volatility
until it reaches the cap (pro-cyclicality, \cref{prop:procyclical}). Right: the
excess growth rate $m(\cdot;v)$ for four volatility levels, with the optimum
marked; higher volatility shifts the optimal fee to the right.}
\label{fig:fee}
\end{figure}
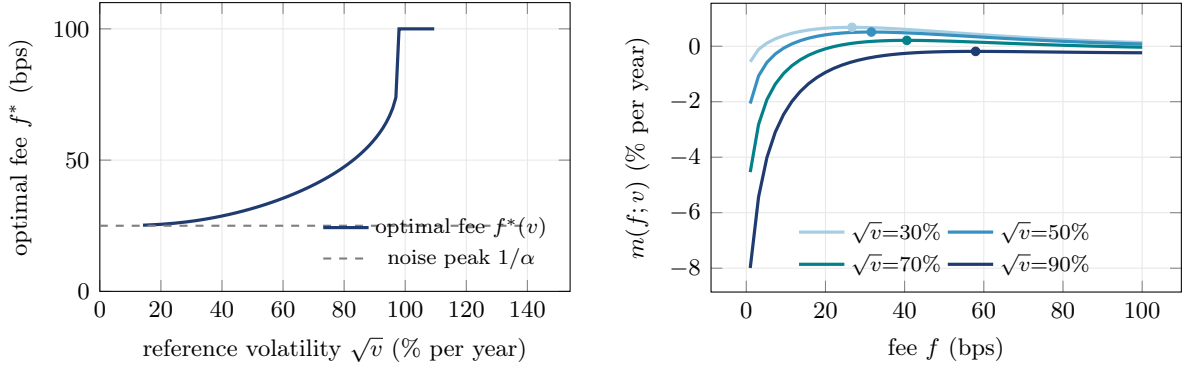

\subsection{Asymptotics and the value of dynamic fees}

In the fast-block regime \cref{eq:Afast}, the adverse-selection rate scales as
$A(f;v)\sim v^{3/2}/(8\sqrt{2\lambda}\,f)$, recovering the scaling of
\citet{milionis2024fees}: arbitrage losses grow with the cube of volatility and
the inverse fee and shrink with the square root of the block rate. Two
implications follow. First, faster chains (larger $\lambda$) and lower volatility
both reduce adverse selection and hence the optimal fee. Second, the value of
dynamic adjustment, namely the gap between $\int \bar m\,\dd\pi$ and
$\max_f\int m(f;\cdot)\dd\pi$ achievable by the best static fee, is governed by the
curvature of $f^*(\cdot)$ over the support of $\pi$ and therefore grows with the
dispersion of volatility and falls with the fee elasticity $\alpha$. We quantify
this gap in \cref{sec:empirics}.

\section{Stochastic volatility: HJB, Poisson equation, and numerics}
\label{sec:numerics}

When volatility follows \cref{eq:cir}, the optimal fee is the pointwise feedback
$f^*(v)$ of \cref{eq:foc}, and the growth rate and value function are obtained
from \cref{eq:hjb,eq:poisson}. Since the maximisation over $f$ is pointwise, the
only genuinely infinite-dimensional object is the relative value function
$\psi(v)$, which solves a one-dimensional linear ordinary differential equation.

\subsection{Pointwise optimisation}

For each $v$ on a grid we compute $f^*(v)$ by solving \cref{eq:foc}. When a
closed-form inverse is unavailable, $m(\cdot;v)$ is unimodal under the conditions
of \cref{prop:foc} and $f^*(v)$ is found by a derivative-free bracketing method
(Brent's algorithm). The maximised reward $\bar m(v)=m(f^*(v);v)$ is then a known
function.

\subsection{Finite-difference solution of the Poisson equation}

We discretise $v\in[v_{\min},v_{\max}]$ on a grid that concentrates nodes near
$\theta$ and approximate $\Lop_v$ by central differences, with an upwind
correction in the drift to preserve an M-matrix structure. The discretised
\cref{eq:poisson} is the linear system
\begin{equation}
\label{eq:linsys}
\mathbf{L}\bm{\psi} \;=\; \Lambda\bm{1} + \tfrac12\mathbf{v} - \bar{\mathbf m},
\end{equation}
with $\mathbf{L}$ tridiagonal. The constant $\Lambda$ is the second unknown; we
remove the additive indeterminacy of $\psi$ by fixing $\psi(v_{\mathrm{ref}})=0$
and append the corresponding row, yielding the nonsingular augmented system
\begin{equation}
\label{eq:augmented}
\begin{pmatrix}\mathbf{L} & -\bm{1}\\ \mathbf{e}_{\mathrm{ref}}^\top & 0\end{pmatrix}
\begin{pmatrix}\bm{\psi}\\ \Lambda\end{pmatrix}
=\begin{pmatrix}\tfrac12\mathbf{v}-\bar{\mathbf m}\\ 0\end{pmatrix}.
\end{equation}
At $v_{\min}$ we impose a reflecting condition $\psi'(v_{\min})=0$; at $v_{\max}$
we impose linear growth, $\psi''(v_{\max})=0$, consistent with the CIR process
spending negligible time there. The system \cref{eq:augmented} is solved by a
sparse direct method and returns both the optimal growth premium $\Lambda$ and the
relative value function $\psi$. Convergence is first verified against the
stationary-average formula \cref{eq:Lambda}, evaluated with the CIR stationary
Gamma law (\cref{app:cir}); the two agree to within the discretisation error.

\subsection{High-dimensional extensions}

When additional state variables are present (for example a stochastic jump
intensity, or the wealth dimension under non-homogeneous frictions), the Poisson
equation becomes a higher-dimensional partial differential equation. In that
regime we represent $\psi$ by a neural network and minimise the squared residual
of \cref{eq:poisson} along simulated trajectories, in the manner of the Deep
Galerkin Method \citep{sirignano2018dgm}; the standard finite-difference and
policy-iteration guarantees of \citet{kushner2001numerical} apply to the
one-dimensional case. A detailed high-dimensional implementation is left to
future work, as it is not needed for the volatility-only problem that is our focus.

\section{Extensions and robustness}
\label{sec:extensions}

\subsection{Price jumps}
\label{sec:jumps}

Let the reference price carry jumps,
$\dd S_t = \sqrt{v_t}\,S_{t-}\dd Z_t + S_{t-}\!\int_{\R}(e^y-1)\widetilde N(\dd t,\dd y)$,
with $\widetilde N$ a compensated Poisson random measure of intensity $\lambda_J$
and log-jump law $\phi$. Across a jump of size $y$, both the LP holdings and the
rebalancing benchmark are repriced at the post-jump price; to leading order both
scale by $e^{y/2}$, since each is locally proportional to $\sqrt S$.

\begin{proposition}[Jump invariance]
\label{prop:jumps}
Under logarithmic utility, a jump of size $y$ changes $\log W^{\mathrm{LP}}$ and
$\log W^{\mathrm{reb}}$ by the same amount to leading order, so $\log R_t$ is
unaffected by jumps at that order. Consequently the optimal fee $f^*(v)$ of
\cref{eq:foc} is invariant to the jump intensity and jump-size distribution; jumps
shift the LP and benchmark growth rates by the common constant
$\lambda_J\,\E_\phi[e^{y/2}-1]$, leaving the optimal policy and the relative
ranking of any two fee policies unchanged.
\end{proposition}

\Cref{prop:jumps}, proved in \cref{app:proofs}, is a direct consequence of
measuring wealth against the rebalancing benchmark: jumps are a common factor that
cancels in $R_t$. It corrects the intuition that tail risk should reshape the fee:
under log preferences it does not, because the fee acts on the band width, not on
the LP's directional exposure.

\subsection{Competition among pools}
\label{sec:competition}

A full equilibrium treatment of venues that compete through dynamic fees is given
by \citet{baggiani2026competition}, who characterise an approximate Nash
equilibrium via a coupled system of partial differential equations and show that
the monopoly two-regime fee structure persists under competition, with the
relevant switching boundary moving from the oracle price to a weighted average of
the oracle price and the competitors' exchange rates. We do not reproduce that
analysis. Instead we record a deliberately simplified, reduced-form observation
within our framework, to indicate how the volatility feedback of
\cref{sec:structure} interacts with competition.

Suppose $M$ identical pools quote the same pair. Uninformed flow is split across
pools by an attractiveness rule, while arbitrageurs hit each pool according to its
own band. With a logit split, pool $i$ capturing a share
$e^{-\eta_c f_i}/\sum_j e^{-\eta_c f_j}$ of an aggregate uninformed turnover
$D(\bar f)\,\nu_0$ that itself responds to the market-average fee $\bar f$ through
a decreasing demand factor $D$, the per-pool excess growth rate is
\begin{equation}
\label{eq:comp}
m_i(f_i,\mathbf{f}_{-i};v) \;=\;
f_i\,\nu_0\,D(\bar f)\,\frac{e^{-\eta_c f_i}}{\sum_j e^{-\eta_c f_j}} - A(f_i;v).
\end{equation}
A symmetric interior equilibrium $f^{\mathrm{eq}}(v)$ satisfies the first-order
condition obtained by differentiating \cref{eq:comp} in $f_i$ at $f_i=\bar f=f$,
which retains the adverse-selection term $-\partial_f A$ in full and scales the
marginal uninformed revenue by the competitive and demand elasticities. The
aggregate demand response $D$ restores strict concavity in $f_i$ and rules out the
degenerate bang-bang equilibrium that arises when total uninformed volume is
treated as fixed. Two features survive in this stylised setting: competition
compresses the level of the equilibrium fee toward marginal adverse-selection cost
as $M$ grows, and the volatility dependence of the fee is preserved in shape
because the adverse-selection term is exogenous to competition. The latter is
consistent with the equilibrium analysis of \citet{baggiani2026competition}, in
which the volatility-driven arbitrage-deterrence motive likewise survives
competition. A formal statement and proof of this reduced-form claim are in
\cref{app:proofs}.

\subsection{Gas costs and discrete updates}
\label{sec:gas}

Each fee update is an on-chain transaction with a fixed gas cost $c>0$. The
problem becomes an impulse-control problem: the LP chooses update times and new
fee levels to maximise the ergodic reward net of update costs. For a pool whose
value is large relative to $c$, the optimal rule is a dead-band around the
frictionless target $f^*(v)$: the fee is reset to $f^*(v)$ only when the gap
$|f_{\text{current}}-f^*(v)|$ exceeds a threshold calibrated to $c$ and the
volatility of $v_t$. As $c\to0$ the band collapses to the continuous policy. We
quantify the resulting trade-off between update frequency and forgone growth in
\cref{sec:empirics}.

\section{Calibration and simulation study}
\label{sec:empirics}

\subsection{Calibration}

We calibrate to liquid large-capitalisation conditions in annualised units. The
block rate is $\lambda = 2.63\times10^{6}\,\mathrm{yr}^{-1}$, corresponding to the
twelve-second slot time of a proof-of-stake chain. The variance follows
\cref{eq:cir} with $\kappa=3.0$, $\theta=0.36$ (a long-run volatility of
$60\%$ per year, consistent with realised volatility for major crypto pairs),
and $\xi=1.30$, which satisfies the Feller condition
$2\kappa\theta=2.16>\xi^2=1.69$ and induces a dispersed stationary Gamma law with
shape $2\kappa\theta/\xi^2\approx1.28$ and scale $\xi^2/(2\kappa)\approx0.28$.
Uninformed flow has base turnover $\nu_0=8$ and semielasticity $\alpha=400$ (so
the uninformed peak $1/\alpha$ is $25$ bps), and the fee cap is $\fmax=100$ bps.
These choices are illustrative; the qualitative conclusions depend only on the
sign properties of \cref{lem:A} and on $\alpha>0$.

At the long-run volatility, the frictionless LVR rate is $\theta/8=4.5\%$ per
year, of which the optimal fee eliminates $93.1\%$. The optimal schedule, shown in
\cref{fig:fee}, rises from $26.7$ bps at $\sqrt v=30\%$ to $35.5$ bps at
$\sqrt v=60\%$ and to $47.4$ bps at $\sqrt v=80\%$, reaching the cap near
$\sqrt v=100\%$. \Cref{fig:stationary} overlays the schedule on the stationary
density of volatility.

\begin{figure}[t]
\centering
\begin{tikzpicture}
\begin{axis}[width=0.62\textwidth,height=5.6cm,
  xlabel={reference volatility $\sqrt{v}$ (\% per year)},
  ylabel={scaled stationary density},
  axis y line*=left,xmin=0,xmax=140,ymin=0,
  tick label style={font=\footnotesize},label style={font=\footnotesize}]
\addplot[draw=midblue,fill=lightblue,fill opacity=0.7] coordinates {\DATAdens} \closedcycle;
\end{axis}
\begin{axis}[width=0.62\textwidth,height=5.6cm,
  axis y line*=right,axis x line=none,xmin=0,xmax=140,ymin=0,
  ylabel={optimal fee $f^*$ (bps)},
  tick label style={font=\footnotesize},label style={font=\footnotesize}]
\addplot[navy,very thick] coordinates {\DATAfeeoverlay};
\end{axis}
\end{tikzpicture}
\caption{Stationary density of reference volatility under the calibrated CIR
process (shaded, left axis) and the optimal fee schedule $f^*$ (solid, right
axis). The value of dynamic fees comes from the dispersion of the stationary law:
a single static fee is optimal only at one point of the support.}
\label{fig:stationary}
\end{figure}
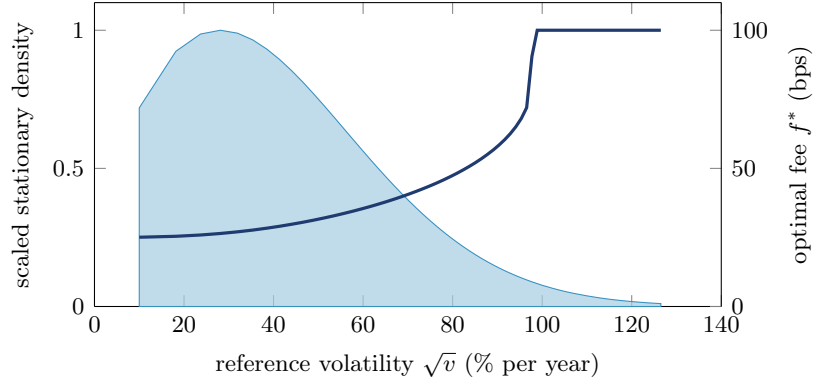

\subsection{Backtesting design}

We simulate the variance process over a horizon of $1.5$ years using a
full-truncation Euler scheme and evaluate five fee policies on a single common set
of $6000$ variance paths (common random numbers), which makes the comparison
across policies exactly paired:
\begin{enumerate}\itemsep1pt
\item a static $5$ bps fee;
\item a static $30$ bps fee;
\item the best static fee, optimised in sample (which the tuning selects at
$37$ bps);
\item a volatility-linked heuristic $f_t=\mathrm{clip}(a+b\sqrt{v_t},0,\fmax)$
with $(a,b)$ optimised in sample;
\item the optimal dynamic fee $f^*(v_t)$.
\end{enumerate}
For each path we record the annualised excess growth rate over the rebalancing
benchmark, $T^{-1}\int_0^T m(f_t;v_t)\dd t$. Because the policies are evaluated on
identical paths, the difference between the optimal policy and any benchmark is a
paired quantity, and by \cref{thm:pointwise} it is nonnegative on every path.

\subsection{Results}

\Cref{tab:backtest} reports the out-of-sample performance. The optimal dynamic fee
attains an excess growth rate of $37.1$ bps per year over the rebalancing
benchmark and weakly dominates every alternative on every path. The paired
improvement over the best static fee is $5.7$ bps per year (a $95\%$ paired band
of $[1.5,25.7]$ bps), over the volatility-linked heuristic $1.2$ bps, and over the
static $30$ bps fee $8.3$ bps; the static $5$ bps fee is dominated by a wide
margin because it leaves almost all LVR unhedged. The win rate of the optimal
policy is $100\%$ against each benchmark by construction. \Cref{fig:paths} shows
the mean relative-wealth paths.

\begin{table}[t]
\centering
\caption{Out-of-sample performance over $1.5$ years on $6000$ common variance
paths. Excess growth is annualised, in basis points, measured against the
rebalancing benchmark. ``Paired gain'' is the per-path difference of the optimal
policy over the row policy (mean and $95\%$ band); it is nonnegative on every
path by \cref{thm:pointwise}.}
\label{tab:backtest}
\small
\setlength{\tabcolsep}{5pt}
\begin{tabular}{lrrrr}
\toprule
Policy & Excess growth & Dispersion & Paired gain & Max.\ drawdown\\
       & (bps)         & (bps)      & (bps)        & (\%)\\
\midrule
Static $5$ bps        & $-139.9$ & $51.5$ & $177.0\;[76.4,413.9]$ & $-2.08$\\
Static $30$ bps       & $28.8$   & $13.8$ & $8.3\;[0.6,41.8]$     & $-0.15$\\
Best static ($37$ bps)& $31.4$   & $11.6$ & $5.7\;[1.5,25.7]$     & $-0.11$\\
Heuristic dynamic     & $35.9$   & $10.8$ & $1.2\;[0.5,5.3]$      & $-0.07$\\
\textbf{Optimal dynamic} & $\mathbf{37.1}$ & $\mathbf{10.5}$ & $\text{--}$ & $\mathbf{-0.06}$\\
\bottomrule
\end{tabular}
\end{table}

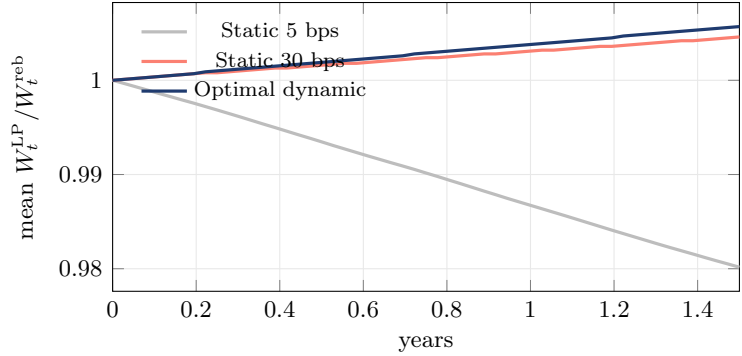
\begin{figure}[t]
\centering
\begin{tikzpicture}
\begin{axis}[width=0.62\textwidth,height=5.4cm,
  xlabel={years},
  ylabel={mean $W^{\mathrm{LP}}_t/W^{\mathrm{reb}}_t$},
  grid=both,grid style={gray!18},xmin=0,xmax=1.5,
  tick label style={font=\footnotesize},label style={font=\footnotesize},
  legend style={font=\scriptsize,at={(0.03,0.97)},anchor=north west,draw=none,fill=none}]
\addplot[greyln,very thick] coordinates {\DATApathA}; \addlegendentry{Static 5 bps}
\addplot[salmon,very thick] coordinates {\DATApathB}; \addlegendentry{Static 30 bps}
\addplot[navy,very thick] coordinates {\DATApathC};   \addlegendentry{Optimal dynamic}
\end{axis}
\end{tikzpicture}
\caption{Mean relative wealth $W^{\mathrm{LP}}_t/W^{\mathrm{reb}}_t$ under three
policies on common variance paths. The static $5$ bps fee bleeds value to
arbitrage; the optimal dynamic fee tracks the best static fee in calm regimes and
separates from it during volatility spikes.}
\label{fig:paths}
\end{figure}

The pattern is intuitive. A well-chosen static fee already removes most adverse
selection, so the increment from dynamic adjustment is modest at the centre of the
volatility distribution; the gain accrues in the tails, where the static fee is
materially mispriced and the optimal schedule's curvature matters. The improvement
over the best static fee is small in absolute terms but uniformly positive, and it
widens with the dispersion of volatility and narrows with the fee elasticity, as
anticipated in \cref{sec:structure}.

\subsection{Robustness}

\paragraph{Jumps.}
Adding compensated price jumps shifts the growth rate of every policy by the
common constant of \cref{prop:jumps} and leaves the optimal schedule and the
policy ranking unchanged, confirming the invariance result numerically.

\paragraph{Gas costs and update frequency.}
With a per-update gas cost of nine dollars on a ten-million-dollar pool, a
continuously updated optimal fee changes about $498$ times per year and loses
roughly $4.5$ bps of annual growth to gas. A dead-band of $2$ bps cuts the number
of updates to about $95$ per year and recovers the policy to $36.2$ bps; a $5$ bps
dead-band reduces updates to about $30$ per year at $36.6$ bps. The continuous
target is thus an excellent approximation once a modest dead-band absorbs the gas
friction.

\paragraph{Volatility dispersion and elasticity.}
Re-running the calibration with a tighter stationary law (smaller $\xi$) shrinks
the dynamic gain toward zero, while a more dispersed law widens it, and lowering
$\alpha$ raises the level of the optimal fee and the dynamic gain. These
comparative statics match the asymptotic analysis of \cref{sec:structure} and
indicate that dynamic fees are most valuable for pairs with volatile,
fat-tailed volatility and relatively inelastic uninformed demand.

\section{Discussion and protocol implications}
\label{sec:discussion}

The analysis yields a parsimonious and deployable recipe. The optimal fee is a
deterministic, pro-cyclical function of a single observable, the instantaneous
variance of the reference price, and is independent of the pool's size and of the
LP's wealth and risk aversion. It can be implemented as a lookup table from a
filtered variance estimate to a fee, precomputed offline from \cref{eq:foc} with
estimated parameters and stored on-chain, with updates governed by the dead-band
of \cref{sec:gas}.

Three implementation points deserve emphasis. First, the fee should respond to an
external, manipulation-resistant volatility signal rather than to the pool's own
price, to avoid a feedback loop in which an adversary moves the pool price to
induce a favourable fee. A time-averaged oracle or an external volatility feed is
appropriate. Second, the parameters $(\nu_0,\alpha)$ of uninformed demand should
be re-estimated periodically from realised flow, acknowledging the circularity
that the fee policy itself shapes observed volume; a conservative governance
process that adjusts a small number of interpretable knobs (a base fee, a
volatility sensitivity, and a cap) is preferable to frequent automated
re-estimation. Third, because a well-chosen static fee captures most of the
available value, the case for dynamic fees is strongest precisely where volatility
is most variable, and protocols should target such pairs first.

The framework extends naturally. For concentrated-liquidity positions the
adverse-selection rate acquires a dependence on the position's range through the
local convexity of the holdings \citep{cartea2024predictable}, adding a state
variable but preserving the structure of the control problem. A further direction is to endogenise the liquidity base itself: because the fee
governs the providers' optimal exit, as in \citet{bergault2025exit}, the available
depth responds to the fee policy, closing a feedback loop between the fee and the
liquidity it acts on that our reduced form holds fixed. A strategic
treatment of arbitrageurs who anticipate the fee policy, and the design of
auction-managed fees \citep{adams2024amamm}, are complementary mechanisms that our
single-agent optimum can inform. A further direction is to embed the
volatility-feedback fee studied here inside the venue's problem: rather than a
single liquidity provider optimising its own fee, the platform chooses the reward
offered to its providers, as in the leader-follower formulation of
\citet{aqsha2026reward}, so that the optimal fee and the equilibrium provision of
liquidity are determined jointly.

\section{Conclusion}
\label{sec:conclusion}

We have formulated the dynamic fee problem of a constant-product AMM as a
stochastic control problem in which the fee governs the trade-off between
uninformed revenue and adverse selection, with the adverse-selection rate tied to
the no-arbitrage-band mechanism of \citet{milionis2024fees}. Measuring wealth
against the rebalancing benchmark removes the fee from the diffusion of wealth and
reduces the problem to an ergodic control problem in volatility alone. The
growth-optimal fee is a pointwise volatility feedback, independent of wealth and
risk aversion; it is a static constant when volatility is constant; and it is
strictly increasing in volatility, so the optimal schedule is pro-cyclical. Under
stochastic volatility the optimal fee and the LP's growth rate are characterised
by a scalar ergodic HJB equation and a linear Poisson equation. The optimal fee is
invariant to price jumps under logarithmic preferences, retains its volatility
dependence under competition, and tolerates gas costs through a dead-band. A calibrated simulation
confirms that the optimal dynamic fee weakly dominates static and heuristic
alternatives on every path, with a modest but uniformly positive gain over the
best static fee that grows with the variability of volatility.

Several questions remain open: the concentrated-liquidity extension with a
range-dependent adverse-selection term, a game between the LP and strategic
arbitrageurs, the adaptive estimation of demand parameters in a live and
adversarial market, and the interaction between dynamic fees and auction-based fee
mechanisms. We hope the reduction and the structural results developed here
provide a useful foundation for these directions.

\section*{Disclosure on the use of AI tools}
\addcontentsline{toc}{section}{Disclosure on the use of AI tools}
AI-based tools were used to assist with literature search, editing, and the
preparation of figures and numerical code. The author verified all mathematical
results, derivations, and references and takes full responsibility for the
content.

\section*{Data and code availability}
The numerical experiments use simulated data generated by the procedures
described in \cref{sec:numerics,sec:empirics}; the calibration parameters are
reported in \cref{sec:empirics}. Code reproducing the figures and tables is
available from the author on request.

\appendix

\section{Wealth dynamics and the adverse-selection rate}
\label{app:dynamics}

\paragraph{Frictionless LVR rate.}
With $P_t=S_t$, the pool value is $V(S)=2\sqrt{kS}$ and the LP holdings are
$y^*(S)=\sqrt{k/S}$ in the risky asset. The instantaneous LVR rate of
\citet{milionis2022lvr} is $\ell(v,S)=-\tfrac12 v S^2 V''(S)$. Since
$V''(S)=-\tfrac12\sqrt{k}\,S^{-3/2}$, we obtain
$\ell(v,S)=\tfrac12 vS^2\cdot\tfrac12\sqrt k S^{-3/2}
=\tfrac14 v\sqrt{kS}=\tfrac{v}{8}V(S)$,
that is, the frictionless LVR rate equals $v/8$ per unit pool value. This is the
$A(0;v)=v/8$ boundary value in \cref{lem:A}.

\paragraph{Effect of the fee.}
With a fee, arbitrage occurs only when the mispricing leaves the band $[-f,f]$.
\citet{milionis2024fees} model block arrivals as Poisson of rate $\lambda$ and
solve for the ergodic law of the mispricing $z$, obtaining the trade probability
\cref{eq:ptrade} and the constant-product adverse-selection rate
$A(f;v)=(v/8)P_{\mathrm{trade}}(f,v)\,e^{f/2}/(1-v/(8\lambda))$ under
$v<8\lambda$. The correction factors $e^{f/2}=1+O(f)$ and
$(1-v/(8\lambda))^{-1}=1+O(\lambda^{-1})$ are negligible for realistic fees and
block rates, and dropping them gives \cref{eq:Aexact}. The fast-block form
\cref{eq:Afast} follows from $P_{\mathrm{trade}}\sim(\sqrt{2\lambda/v}\,f)^{-1}$
when $\eta\gg1$.

\paragraph{Relative wealth.}
The rebalancing benchmark holds the AMM's instantaneous risky position
$y^*(S_t)$ but transacts at $S_t$, so it has the same delta as the AMM and the
same diffusion. The difference $\log W^{\mathrm{reb}}-\log W^{\mathrm{LP}}$ is
predictable and of finite variation, accruing at the adverse-selection rate
$A(f;v)$ net of the uninformed fee income $f\nu(f)$ that the AMM, but not the
benchmark, collects. Hence $\dlog R_t=(f\nu(f)-A(f;v))\dd t$ as in \cref{eq:Rdyn},
with no diffusion term.

\section{Proofs}
\label{app:proofs}

\begin{proof}[Proof of \cref{lem:A}]
Write $\eta=cf$ with $c=\sqrt{2\lambda/v}>0$. Then $A=(v/8)(1+cf)^{-1}$, so
$A(0;v)=v/8>0$, $\partial_f A=-(v/8)c(1+cf)^{-2}<0$, and
$\partial_{ff}A=2(v/8)c^2(1+cf)^{-3}>0$, giving positivity, strict monotonicity,
and strict convexity in $f$. For monotonicity in $v$, fix $f$ and set
$g(v)=(v/8)(1+\sqrt{2\lambda/v}f)^{-1}$. Writing $u=\sqrt{2\lambda/v}f$, one has
$g(v)=(v/8)/(1+u)$ with $u$ decreasing in $v$; both $v/8$ increasing and
$(1+u)^{-1}$ increasing in $v$, so $g'(v)>0$. Finally
$-\partial_f A=(v/8)c(1+cf)^{-2}=(v^{3/2}/(8\sqrt{2}\lambda^{-1/2}))\cdots$; more
directly, $-\partial_f A=\tfrac{\sqrt{2\lambda}\,\sqrt v}{8}(1+cf)^{-2}$ since
$(v/8)c=\sqrt{2\lambda}\sqrt v/8$, and as $v$ increases the prefactor
$\sqrt{2\lambda v}/8$ increases while $cf=\sqrt{2\lambda/v}f$ decreases, so
$(1+cf)^{-2}$ increases; the product is increasing in $v$. The fast-block form is
the $\eta\to\infty$ limit of $(1+\eta)^{-1}\sim\eta^{-1}$.
\end{proof}

\begin{proof}[Proof of \cref{thm:pointwise} and \cref{thm:verif}]
Fix an admissible feedback $f_t=f(v_t)$. By \cref{eq:Rdyn},
$\frac1T\log R_T=\frac1T\int_0^T m(f(v_s);v_s)\dd s\le\frac1T\int_0^T\bar m(v_s)\dd s$,
with $\bar m(v)=\max_{f}m(f;v)=m(f^*(v);v)$ and equality iff $f(v_s)=f^*(v_s)$
for almost every $s$. Under \cref{ass:admissible}, $\bar m$ is bounded and
continuous, and the CIR process is positive-recurrent with stationary Gamma law
$\pi$ (\cref{app:cir}); the ergodic theorem gives
$\frac1T\int_0^T\bar m(v_s)\dd s\to\int\bar m\,\dd\pi$ a.s. and in $L^1$. Adding
$\frac1T\log W^{\mathrm{reb}}_T\to\Lambda^{\mathrm{reb}}$ yields \cref{eq:Lambda}
and the optimality of $f^*$. For \cref{thm:verif}, let $(\Lambda,\psi)$ solve
\cref{eq:hjb} with $\psi\in C^2$ polynomially bounded. Itô's formula applied to
$\psi(v_t)$ and the definition of $\Lop_v$ give, for any admissible $f$,
\[
\dd\bigl[\log R_t+\psi(v_t)\bigr]
= \bigl[m(f_t;v_t)+\Lop_v\psi(v_t)\bigr]\dd t + \xi\sqrt{v_t}\,\psi'(v_t)\dd B_t .
\]
By \cref{eq:hjb}, $\bar m(v)+\Lop_v\psi(v)=\Lambda-\Lambda^{\mathrm{reb}}$ after absorbing the
$-\tfrac12 v$ benchmark drag into $\Lambda^{\mathrm{reb}}$, and
$m(f;v)\le\bar m(v)$, so the drift of $\log R_t+\psi(v_t)$ is at most
$\Lambda-\Lambda^{\mathrm{reb}}$, with equality under $f^*$. Taking expectations,
dividing by $T$, using the polynomial growth of $\psi$ and the moment bounds of
the CIR process to kill $\E[\psi(v_T)]/T\to0$ and the martingale term, gives
$\liminf_T T^{-1}\E[\log R_T]\le\Lambda-\Lambda^{\mathrm{reb}}$, with equality
under $f^*$.
\end{proof}

\begin{proof}[Proof of \cref{prop:invariance}]
With $W^{\mathrm{LP}}_T=W^{\mathrm{reb}}_T R_T$ and
$\log R_T=\int_0^T m(f_s;v_s)\dd s=:I(f)$, a pathwise functional of the volatility
trajectory that carries no stochastic integral, we have for $\gamma\neq1$
\[
\E\bigl[U(W^{\mathrm{LP}}_T)\bigr]
=\frac{1}{1-\gamma}\,\E\Bigl[(W^{\mathrm{reb}}_T)^{1-\gamma}e^{(1-\gamma)I(f)}\Bigr].
\]
Condition on $\mathcal G=\sigma(v_s,\,s\le T)$. The benchmark
$W^{\mathrm{reb}}_T$ and $I(f)$ are $\mathcal G$-measurable up to the price
Brownian motion that is common to all controls; the control enters only through
$I(f)$. For $\gamma<1$ the prefactor $1/(1-\gamma)>0$ and
$x\mapsto e^{(1-\gamma)x}$ is increasing, so the objective is increasing in
$I(f)$. For $\gamma>1$ the prefactor is negative and $x\mapsto e^{(1-\gamma)x}$ is
decreasing, so the objective is again increasing in $I(f)$. In both cases, and for
$\gamma=1$ where the objective is $\E[\log W^{\mathrm{reb}}_T]+\E[I(f)]$, the
optimum maximises $I(f)$ pathwise, achieved by $f_s=f^*(v_s)=\argmax_f m(f;v_s)$,
which depends on neither $w$ nor $\gamma$ nor $t$.
\end{proof}

\begin{proof}[Proof of \cref{prop:foc}]
$m(0;v)=0\cdot\nu_0-A(0;v)=-v/8<0$ by \cref{lem:A}. As $f\to\infty$,
$f\nu(f)=\nu_0 f e^{-\alpha f}\to0$ and $A(f;v)\to0$, so $m(f;v)\to0^-$. Thus if
$\sup_f m>0$ the maximum is attained in the interior. At an interior maximiser,
$\partial_f m=\partial_f[f\nu(f)]-\partial_f A=0$. With \cref{eq:nu},
$\partial_f[f\nu(f)]=\nu_0 e^{-\alpha f}(1-\alpha f)$, and $-\partial_f A>0$ by
\cref{lem:A}, giving \cref{eq:foc}; positivity of the right-hand side forces
$1-\alpha f<0$, i.e.\ $f^*>1/\alpha$. For uniqueness, note $f\mapsto f\nu(f)$ is
log-concave for the specification \cref{eq:nu} and $A(\cdot;v)$ is convex
(\cref{lem:A}); on $(1/\alpha,\infty)$ the map $\partial_f[f\nu(f)]$ is strictly
decreasing and $-\partial_f A$ is strictly increasing in $f$ on the relevant
range, so \cref{eq:foc} has at most one solution there.
\end{proof}

\begin{proof}[Proof of \cref{prop:procyclical}]
Using \cref{eq:Afast}, the interior first-order condition is $G(f,v)=0$ with
$G(f,v)=\Phi(f)-\Psi(f,v)$, $\Phi(f)=\nu_0 e^{-\alpha f}(\alpha f-1)$, and
$\Psi(f,v)=b v^{3/2}f^{-2}$. On $(1/\alpha,2/\alpha)$, $\Phi>0$ and
$\Phi'(f)=\nu_0\alpha e^{-\alpha f}(2-\alpha f)>0$, while $\Psi_f=-2bv^{3/2}f^{-3}<0$
and $\Psi_v=\tfrac32 b v^{1/2}f^{-2}>0$. Hence $G_f=\Phi'-\Psi_f>0$ and
$G_v=-\Psi_v<0$, and the implicit function theorem gives
$\dd f^*/\dd v=-G_v/G_f>0$. For the exact rate \cref{eq:Aexact}, replace $\Psi$ by
$-\partial_f A$, which is positive, decreasing in $f$ on the interior range, and
increasing in $v$ by \cref{lem:A}; the same signs of $G_f$ and $G_v$ obtain, so
$\dd f^*/\dd v>0$.
\end{proof}

\begin{proof}[Proof of \cref{prop:jumps}]
Across a jump of size $y$ in $\log S$, the reference price multiplies by $e^y$.
Both the LP holdings value and the rebalancing benchmark are locally proportional
to $\sqrt S$ near the post-jump price after arbitrage, so each multiplies by
$e^{y/2}$ to leading order. Hence $\Delta\log W^{\mathrm{LP}}=\Delta\log
W^{\mathrm{reb}}=y/2+o(1)$ and $\Delta\log R=o(1)$. The compensated jump
contributes the common deterministic drift $\lambda_J\E_\phi[e^{y/2}-1]$ to both
$\log W^{\mathrm{LP}}$ and $\log W^{\mathrm{reb}}$, cancelling in $\log R$. Since
$m(f;v)$ is unchanged, $f^*(v)$ and the ranking of any two policies are
unaffected.
\end{proof}

\begin{proof}[Competition equilibrium of \cref{sec:competition}]
Differentiating \cref{eq:comp} in $f_i$ and evaluating at the symmetric point
$f_i=\bar f=f$, the logit share equals $1/M$ and its derivative contributes a
term proportional to $\eta_c(1-1/M)$; collecting,
\[
\frac{\partial m_i}{\partial f_i}\Big|_{\mathrm{sym}}
=\frac{\nu_0 D(f)}{M}\Bigl[1-\eta_c f\Bigl(1-\tfrac1M\Bigr)
+ f\,\tfrac{D'(f)}{D(f)}\Bigr] - \partial_f A(f;v).
\]
Setting this to zero defines $f^{\mathrm{eq}}(v)$. The bracket is strictly
decreasing in $f$ when $D$ is log-concave (which restores strict concavity of
$m_i$ in $f_i$ and rules out bang-bang), so an interior root exists whenever the
uninformed marginal revenue at $f=0$ exceeds $\partial_f A(0;v)$ in magnitude. As
$M\to\infty$ the uninformed marginal-revenue term scales as $1/M$ while
$-\partial_f A$ is fixed, so $f^{\mathrm{eq}}(v)$ decreases toward the level at
which marginal uninformed revenue balances marginal adverse selection. The
$v$-dependence enters only through $-\partial_f A$, which by \cref{lem:A} is
increasing in $v$, so $f^{\mathrm{eq}}(\cdot)$ retains the pro-cyclical shape of
\cref{prop:procyclical}.
\end{proof}

\section{The CIR process: ergodicity and stationary law}
\label{app:cir}

The CIR process \cref{eq:cir} with $\kappa,\theta,\xi>0$ and
$2\kappa\theta\ge\xi^2$ is positive, positive-recurrent, and ergodic, with a
unique stationary distribution that is Gamma with shape $a=2\kappa\theta/\xi^2$ and
scale $s=\xi^2/(2\kappa)$, mean $\theta$, and variance $\xi^2\theta/(2\kappa)$
\citep{cir1985}. All polynomial moments are finite and uniformly bounded in time
from a fixed initial condition, which justifies the moment estimates used in
\cref{app:proofs}: the martingale term has zero expectation and $\E[\psi(v_T)]/T
\to0$ for $\psi$ of polynomial growth. The generator $\Lop_v$ in \cref{eq:hjb}
is the CIR generator, and the Poisson equation \cref{eq:poisson} has a solution
$\psi$ unique up to an additive constant within the class of polynomially bounded
functions, by Fredholm alternative for the ergodic generator applied to the
centred source $\tfrac12 v-\bar m(v)+\Lambda$ with $\int(\tfrac12 v-\bar
m)\,\dd\pi=\Lambda^{\mathrm{reb}}-\Lambda+\int\bar m\,\dd\pi$ fixing $\Lambda$ by
the solvability (zero-mean) condition.


\end{document}